\documentclass{article}

\usepackage[compact]{titlesec}

\usepackage{fancyhdr}
\usepackage{geometry}
\usepackage{epsfig,endnotes,url,color,subfigure}
\usepackage{array}
\usepackage{multirow}
\usepackage{xspace}
\usepackage{verbatim}

\usepackage[english]{babel}
\usepackage{blindtext}

\usepackage{amsthm}
\usepackage{mathpazo} % math & rm
\usepackage{amsmath}

\newcommand{\projecttitle}{Slider\xspace}

\newtheorem{theorem}{Theorem}

\newcommand{\ninput}{\ensuremath{n_i}\xspace}
\newcommand{\noutput}{\ensuremath{n_O}\xspace}
\newcommand{\nmap}{\ensuremath{N_M}\xspace}
\newcommand{\nred}{\ensuremath{N_R}\xspace}
\newcommand{\ncon}{\ensuremath{N_C}\xspace}
\newcommand{\nconAppend}{\ensuremath{N_{CA}}\xspace}
\newcommand{\nconFix}{\ensuremath{N_{CF}}\xspace}
\newcommand{\nconVariable}{\ensuremath{N_{CV}}\xspace}

\newcommand{\nkeysetmap}{\ensuremath{n_{mk}}}
\newcommand{\nkeyvaluemap}{\ensuremath{n_{m}}}

\begin{document}

\date{}
\title{Asymptotic Analysis of Self-Adjusting Contraction Trees} 
\author{Pramod Bhatotia\\bhatotia@mpi-sws.org}

\maketitle
\thispagestyle{empty}
\pagestyle{empty}

In this report, we analyze the asymptotic efficiency of self-adjusting contraction trees proposed as part of the \projecttitle project~\cite{slider, contraction-tree}.  Self-adjusting contraction trees are used for incremental computation~\cite{Bhatotia15,  ithreads, shredder, incApprox}. Our analysis extends the asymptotic efficiency analysis of Incoop~\cite{incoop-hotcloud, incoop}. We consider two
different runs: the {\em initial run} of an \projecttitle computation,
where we perform a computation with some input $I$, and a second
run for {\em dynamic update} where we change the input from $I$ to
$I'$ and perform the same computation with the new input.  In the
common case, we perform a single initial run followed by many dynamic
updates.

For the initial run, we define the {\em overhead} as the slowdown of
\projecttitle compared to a conventional implementation of MapReduce
such as with Hadoop.  We show that the overhead depends on
communication costs and, if these are independent of
the input size, which they often are, then it is also constant.  Our
experiment evaluation confirms that the overhead is relatively
small.  We show that dynamic updates are dominated by the time it
takes to execute fresh tasks that are affected by the changes to the
input data, which, for a certain class of computations and small
changes, is logarithmic in the size of the input.

In the analysis, we use the following terminology to refer to the
three different types of computational tasks that form an
\projecttitle~computation: Map tasks, Self-adjusting balanced tree (applications
of the Combiner function for three different modes of operation for sliding-window computations), and Reduce tasks.

Our bounds depend on the total number of map tasks, written
\nmap, and the total number of reduce tasks written \nred. In addition, we also
take in account the total number of stages in self-adjusting balanced tree,
denoted as \ncon. We write
\ninput and \noutput to denote the total size of the input and output
respectively, \nkeyvaluemap~to denote the total number of key-value
pairs output by the Map phase, and $\nkeysetmap$ to denote the set of
distinct keys emitted by the Map phase. The number of stages in self-adjusting
balanced tree is a property of sliding-window computation mode: append-only
$(\nconAppend  = O(\nkeysetmap))$, fixed-width window slides
$(\nconFix = O(\nkeysetmap \cdot \lceil log_2(buckets) \rceil)) $, and
variable-width window slides$((\nconVariable) = \lceil  O(\nkeysetmap \cdot
 log_2(\nmap) \rceil))$.

For our time bounds, we will additionally assume that each Map,
Combine, and Reduce function performs work that is asymptotically
linear in the size of their inputs.  Furthermore, we will assume that
the Combine function is {\em monotonic}, i.e., it produces an output
that is no larger than its input. This assumption is satisfied in
most applications, because Combiners often reduce the size of
the data (e.g., a Combine function to compute the sum of values takes
multiple values and outputs a single value).

\begin{theorem}[Initial Run:Time and Overhead]
\label{thm:time-initial}
\label{thm:time-overhead}
~\\
Assuming that Map, Combine, and Reduce functions take time
asymptotically linear in their input size and that Combine functions
are monotonic, total time for performing an incremental MapReduce
computation in \projecttitle with an input of size \ninput, where
\nkeyvaluemap key-value pairs are emitted by the Map phase is
$O(\nmap + (\nred + \ncon)) = O(\ninput + \nkeyvaluemap)$. This
results in an overhead of $O(\ncon) = O(\nconAppend || \nconFix ||
\nconVariable$) over conventional MapReduce.
\end{theorem}
\begin{proof}
The number of Map and Reduce tasks in a particular job can be derived
from the input size and the number of distinct keys that are emitted
by the Map function: the Map function is applied to \emph{splits} that
consist of one or more input chunks, and each application of the Map
function is performed by one Map task. Hence, the number of Map tasks
\nmap is in the order of input size $O(\ninput)$.  In the Reduce
phase, each Reduce task processes all previously emitted key-value
pairs for at least one key, which results in at most $\nred =
\nkeysetmap$ reduce tasks.  To bound the number of self-adjusting balanced tree,
we note that the tree leaves are the
output data chunks of the Map phase, whose internal nodes each has at
least two children. Since there are at most \nkeyvaluemap pairs output
by the Map phase, the total number of reduce tasks is bounded by
\nkeyvaluemap.  Hence the total number of stages in self-adjusting balanced
tree is bounded by $\ncon \in O(\nkeyvaluemap)$.  Since the number of reduce
tasks is bounded by $\nkeysetmap \le \nkeyvaluemap$ , the total number of tasks
is $O(\ninput + \nkeyvaluemap)$.  
\end{proof}

\begin{theorem}[Initial Run: Space]
\label{thm:space-initial}
Total storage space for performing an \projecttitle computation with
an input of size \ninput, where \nkeyvaluemap key-value pairs are
emitted by the Map phase, and where Combine is monotonic is $O(
\nkeyvaluemap)$.
\end{theorem}
\begin{proof}
\projecttitle~requires additional storage space for storing the intermediatery
output of the self-adjusting balanced tree.  Since \projecttitle only keeps data
from the most
recent run (initial or dynamic run), we use storage for
remembering only the task output from the most recent run.  The output size
of the map tasks is bounded by \nkeyvaluemap.  With monotonic Combine
functions, the size of the output of Combine tasks is bounded by
$O(\nkeyvaluemap)$.
\end{proof}

\begin{theorem}[Dynamic Update: Space and Time]
~\\
In \projecttitle, a dynamic update requires time, where F where
$F$ is the set of changed or new (fresh) \emph{Map}, \emph{Combiner}, and
\emph{Reduce} tasks, is 
\[
O\left(\sum\limits_{a\in F} t(a) \right).
\]
The total storage requirement is the same as an initial
run.
\end{theorem}
\begin{proof}
Consider \projecttitle performing an initial run with input $I$ and
changing the input to $I'$ and then performing a subsequent run
(dynamic update). During the dynamic update, tasks with the same type
and input data will re-use the memoized result of the previous runs,
avoiding recomputation.  Thus, only the fresh tasks need to be
executed, which takes $O\left(\sum\limits_{a\in F} t(a)\right)$, where
$F$ is the set of changed or new (fresh) \emph{Map}, \emph{Contract}
and \emph{Reduce} tasks, respectively, and $t(\cdot)$ denotes the
processing time for a given task.
\end{proof}

In the common case, we expect the execution of fresh tasks to
dominate the time for dynamic updates. The time for dynamic update is therefore
likely to be determined by the number of fresh tasks that are created
as a result of a dynamic change.  It is in general difficult to bound
the number of fresh tasks, because it depends on the specifics of the
application.  As a trivial example, consider, inserting a single
key-value pair into the input.  In principle, the new pair can force the
Map function to generate a very large number of new key-value pairs,
which can then require performing many new reduce tasks.  In many
cases, however, small changes to the input lead only to small changes in
the output of the Map, Combine, and Reduce functions, e.g., the Map
function can use one key-value pair to generate several new pairs, and
the Combine function will typically combine these, resulting in a
relatively small number of fresh tasks.
As a specific case, assume that the Map function generates $k$
key-value pairs from a single input record, and that the Combine
function monotonically reduces the number of key-value pairs.

\begin{theorem}[Number of Fresh Tasks]
\label{them:bound-fresh}
If the Map function generates $k$ key-value pairs from a single input
record, and  the Combine function is monotonic, then the number of
fresh tasks is at most $O(k\log{\nkeyvaluemap} + k)$.
\end{theorem}
\begin{proof}
At most $k$ combine at each level of the self-adjusting
tree will be fresh, and $k$ fresh reduce tasks will be needed. Since
the depth of the contraction tree is $\nkeyvaluemap$, the total number
of fresh tasks will therefore be $O(k\log{\nkeyvaluemap} + k) =
O(k\log{\nkeyvaluemap})$.
\end{proof}

Taken together the last two theorems suggest that small changes to
data will lead to the execution of only a small number of fresh tasks, and based
on the tradeoff between the memoization costs and the cost of
executing fresh tasks, speedups can be achieved in practice.

\bibliographystyle{abbrv}

\bibliography{main}

\end{document}